\documentclass[12pt,english]{article}
\usepackage{e-jc}
\specs{P1.35}{22(1)}{2015}

\usepackage{amsthm,amsmath,amssymb}
\usepackage{graphicx}
\newcommand{\arxiv}[1]{\href{http://arxiv.org/abs/#1}{\texttt{arXiv:#1}}}

\usepackage{float}
\usepackage{booktabs}
\usepackage{fixltx2e}
\usepackage{standalone}
\usepackage{tikz}
\usepackage{subcaption}

\theoremstyle{plain}
\newtheorem*{theorem*}{Theorem}

\AtBeginDocument{

}

\begin{document}

\title{\bf On limits of dense packing of equal spheres in a cube}
\author{Milos Tatarevic\\
\small Alameda, CA 94501, U.S.A.\\
\small\tt milos.tatarevic@gmail.com\\
}

\date{\dateline{Oct 8, 2013}{Jan 27, 2015}{Feb 16, 2015}\\
\small Mathematics Subject Classifications: 52C17, 05B40}

\maketitle

%%%%%%%%%%%%%%%%%%%%%%%%%%%%%%

\begin{abstract}
We examine packing of $n$ congruent spheres in a cube when $n$ is
close but less than the number of spheres in a regular cubic close-packed
(ccp) arrangement of $\lceil p^{3}/2\rceil$ spheres. For this
family of packings, the previous best-known arrangements were usually
derived from a ccp by omission of a certain number of spheres without
changing the initial structure. In this paper, we show that better
arrangements exist for all $n\leq\lceil p^{3}/2\rceil-2$. We introduce
an optimization method to reveal improvements of these packings, and
present many new improvements for $n\leq4629$.
\end{abstract}

%%%%%%%%%%%%%%%%%%%%%%%%%%%%%%

\section{Introduction}

We consider the problem of finding the densest packings of congruent, non-overlapping, spheres in a cube.
Equivalently, we can search for an arrangement of points inside a unit cube so that the minimum
distance between any two points is as large as possible.
The maximum separation distance of $n$ points in $[0,1]^{3}$ we denote by $d_{n}$.

To our knowledge, the optimality of $d_{n}$ is proved for $n=2,3,4,5,6,8,9$
\cite{schaer}, $n=10$ \cite{schaer_10} and $n=14$ \cite{joos}.
Optimality of $d_{n}$ was conjectured for an infinite family of packings
where $\lceil p^{3}/2\rceil$ spheres are arranged in a cubic
close-packed (ccp) structure \cite{goldberg}. We denote by $g(p)=\lceil p^{3}/2\rceil$
the number of spheres in these packings, with a maximum separation
distance denoted by $d'_{p}=\sqrt{2}/\!\left(p-1\right)$.

In this paper, we examine arrangements when $n$ is close, but less
than $g(p)$. For this family of packings $d_{n}=d'_{p}$ is often
assumed, to mean that the densest known arrangements are derived from
ccp by omission of a certain number of spheres without changing the
initial structure. Limiting values though were not provided. It was
conjectured that $d_{n}$ is constant in the range $12\leq n\leq14$
\cite{goldberg} and $29\leq n\leq32$ \cite{gensane}. A better
packing was found for $n=12$ \cite{gensane}. Similarly, previous
search results showed the same trend for $60\leq n\leq63$, $103\leq n\leq108$,
$\ldots$, $817\leq n\leq864$ \cite{wenqi,locatelli_de,packomania}.

In Section 2 we show that most of the listed packings can be improved
by proving that $d_{n}>d'_{p}$ for all $n\leq g(p)-2$. We also provide
a lower bound for these improvements. In Section 3, we introduce an
optimization method and improve the lower bound for
$4\leq p\leq21$. We show that the described procedure can be used as a good packing method
when $n$ is slightly smaller than $g(p)$. We run search to
determine improvements for other packings when $n=g(p)-r$,
for $3\leq r\leq6$, $r<p$ and $4\leq p\leq12$.

%%%%%%%%%%%%%%%%%%%%%%%%%%%%%%

\section{Existence of improved packings}

To simplify our notation, we will assume that the radius of all spheres
in the packing is 1 and that our task is to determine the smallest
size of a cube that contains all spheres.

We denote the family of all finite
sets of points such that the distance between any two points is at least 2 by 
\[
\mathcal{F}=\left\{ S\subset\mathbb{R}^{3}:\Vert s_{1}-s_{2}\Vert \geq 2\text{ for all distinct } s_{1},s_{2}\in S \right\}.
\]

Let $S_{n}=\left\{ s_{1,}\ldots,s_{n}\right\} \in\mathcal{F}$ and
let $D_{c}(s_{i},s_{j})$ be the Chebyshev distance between any two
points $s_{i},s_{j}\in S_{n}$. The smallest edge length of a cube,
with edges parallel to the axes, such that it contains all points
$S_{n}$ is equal to 
\[
D(S_{n})=\max\left\{ D_{c}(s_{i},s_{j}):1\leq i<j\leq n\right\}.
\]

We notice that the maximum separation distance $d_{n}$ can also be
given as 
\[
d_{n}=\max\frac{2}{D(S_{n})}.
\]

\begin{theorem*}
The maximum separation distance of $n$ points contained in a closed
region bounded by a unit cube is larger than $\sqrt{2}/\left(p-1\right)$
for all $n\leq\lceil p^{3}/2\rceil-2$.\end{theorem*}

\begin{proof}
Let $C_{p}$ be a closed region bounded by
a cube with an edge length $D_{p}=2/d'_{p}=(p-1)\sqrt{2}$, where $C_{p}$ is defined by
$C_{p}=\left[0,D_{p}\right]^{3}$.
Let $G_{p}$ be a set of $g(p)$ sphere centers in a ccp arrangement, such that
$G_{p}\in\mathfrak{\mathcal{F}}$,
$G_{p}\subset C_{p}$ and $\langle0,0,0\rangle\in G_{p}$ (see Figure 1).

\begin{figure}[ht!]
\centering
\begin{subfigure}{.45\textwidth}
  \centering
  \includegraphics[width=.9\linewidth]{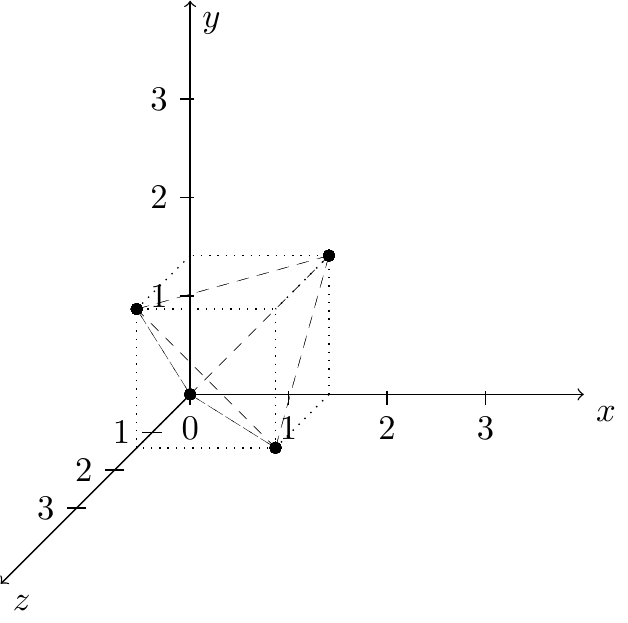}
  \caption{$G_{2}$}
  \vspace*{10mm}
\end{subfigure}%
\begin{subfigure}{.45\textwidth}
  \centering
  \includegraphics[width=.9\linewidth]{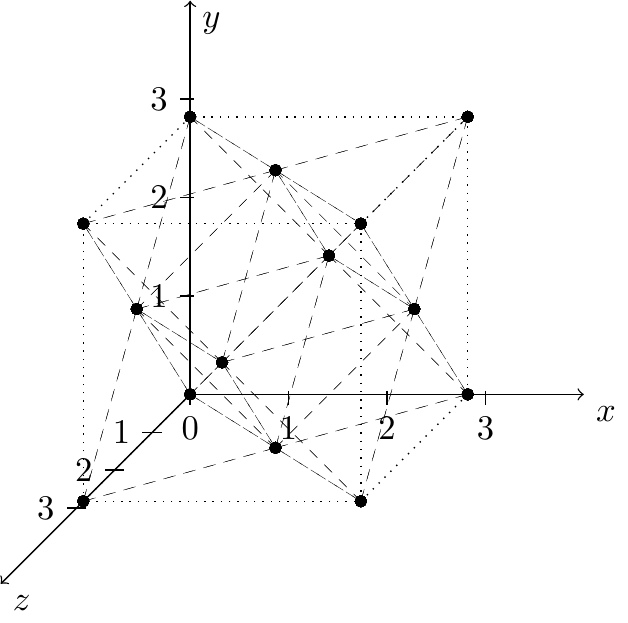}
  \caption{$G_{3}$}
  \vspace*{10mm}
\end{subfigure}
\begin{subfigure}{.6\textwidth}
  \centering
  \includegraphics[width=.9\linewidth]{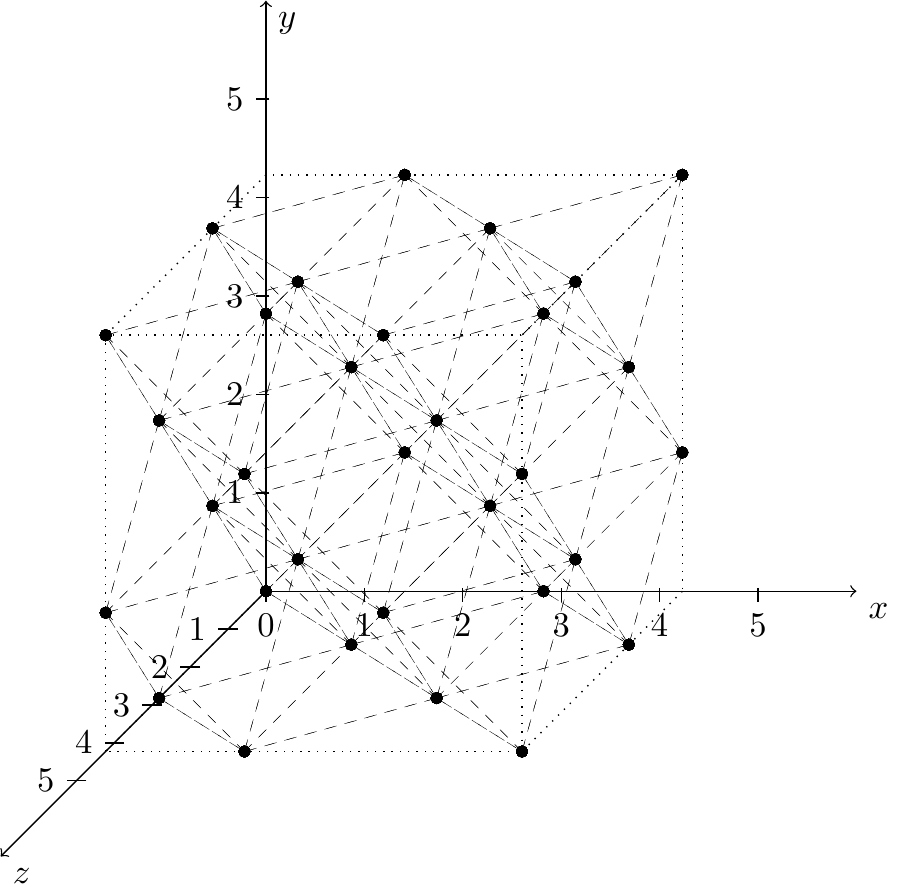}
  \caption{$G_{4}$}
  %\vspace*{10mm}
\end{subfigure}
\caption{An illustration of arrangements $G_{2}$, $G_{3}$ and $G_{4}$.}
\end{figure}

For two sets of points $A,B\in\mathfrak{\mathcal{\mathcal{F}}}$ let
\[
h(A,B)=\min\left\{ \Vert a-b\Vert-2:a\in A,b\in B\right\} 
\]
and let $L_{p}=G_{p}\setminus G_{p-1}$ (see Figure 2(a) for an example). We denote the improved packing
of $g(p)-2$ points by $P_{p}=\left\{ s_{1},\ldots,s_{g(p)-2}\right\} \in\mathcal{\mathcal{F}}$
such that
\begin{equation}
P_{p}\subset C_{p}\label{eq:p_subset},
\end{equation}
\begin{equation}
h(P_{p},L_{p+1})>0\label{eq:p_separation}.
\end{equation}
Equations (\ref{eq:p_subset}) and (\ref{eq:p_separation}) directly
imply that if these statements are true, then $D(P_{p})<D_{p}$.

\begin{figure}[ht!]
\centering
\begin{subfigure}{.5\textwidth}
  \centering
  \includegraphics[width=.9\linewidth]{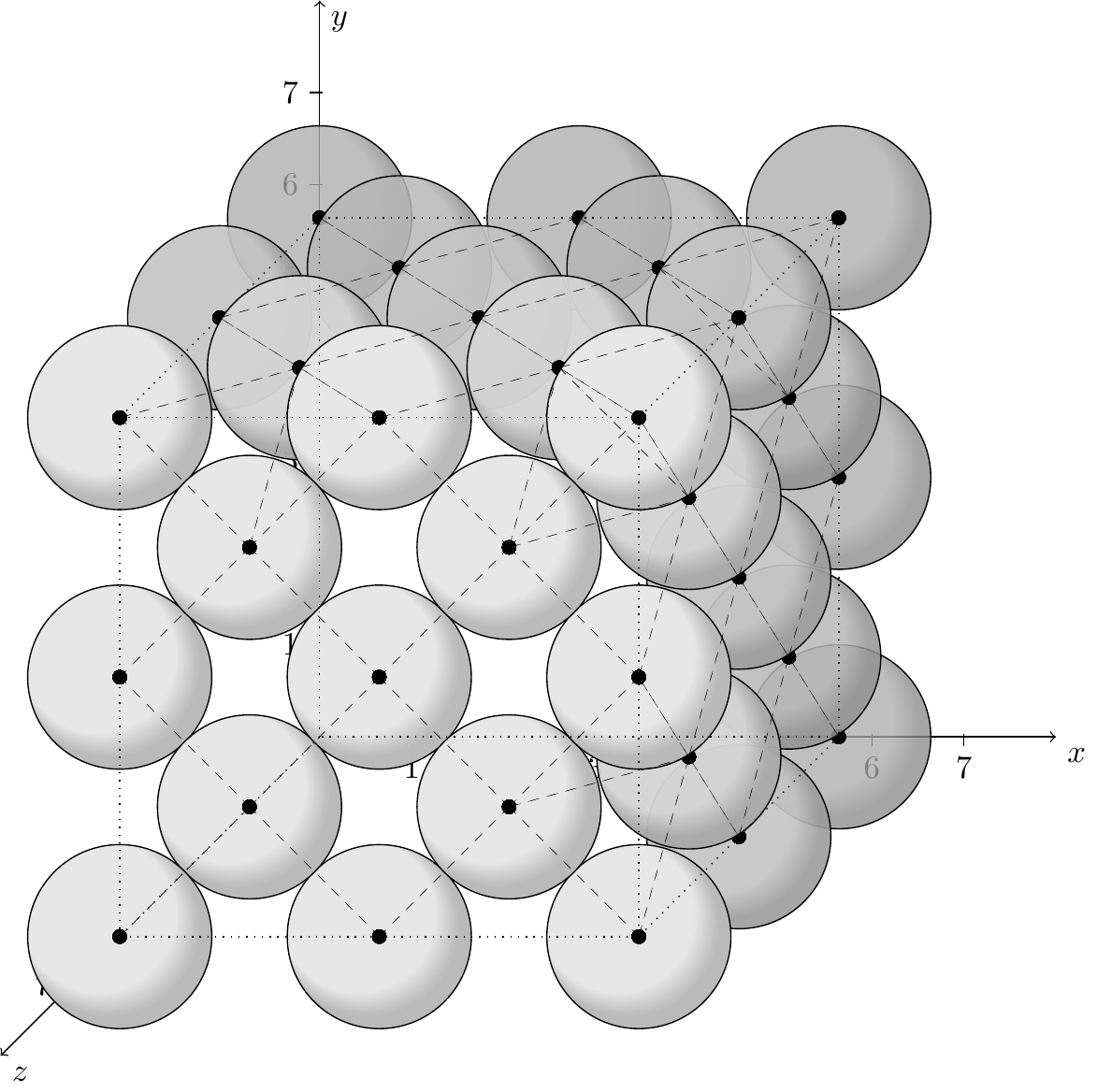}
  \caption{$L_{5}$}
  \vspace*{10mm}
\end{subfigure}%
\begin{subfigure}{.5\textwidth}
  \centering
  \includegraphics[width=.9\linewidth]{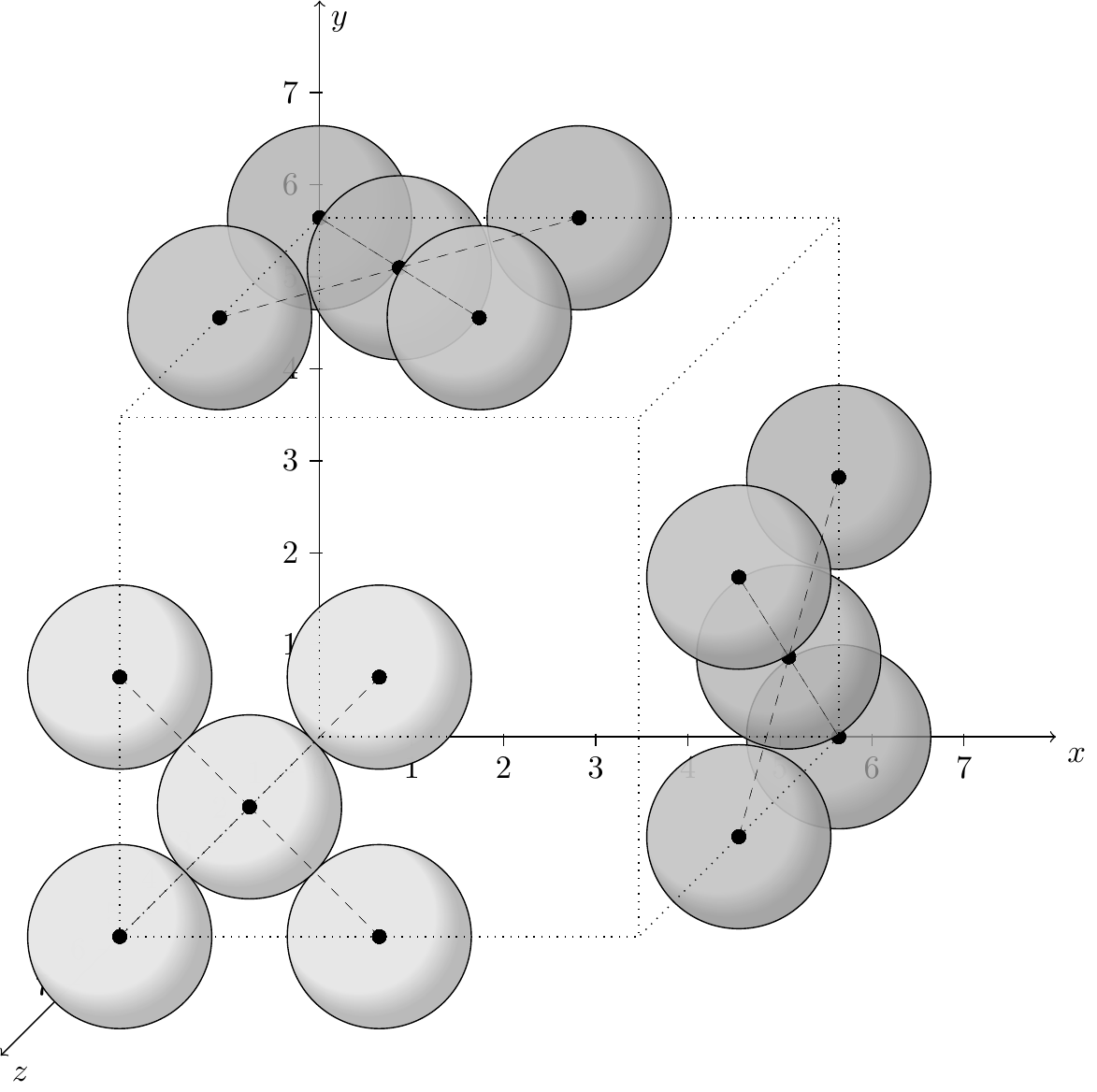}
  \caption{$L_{5,1}$}
  \vspace*{10mm}
\end{subfigure}
\begin{subfigure}{.5\textwidth}
  \centering
  \includegraphics[width=.9\linewidth]{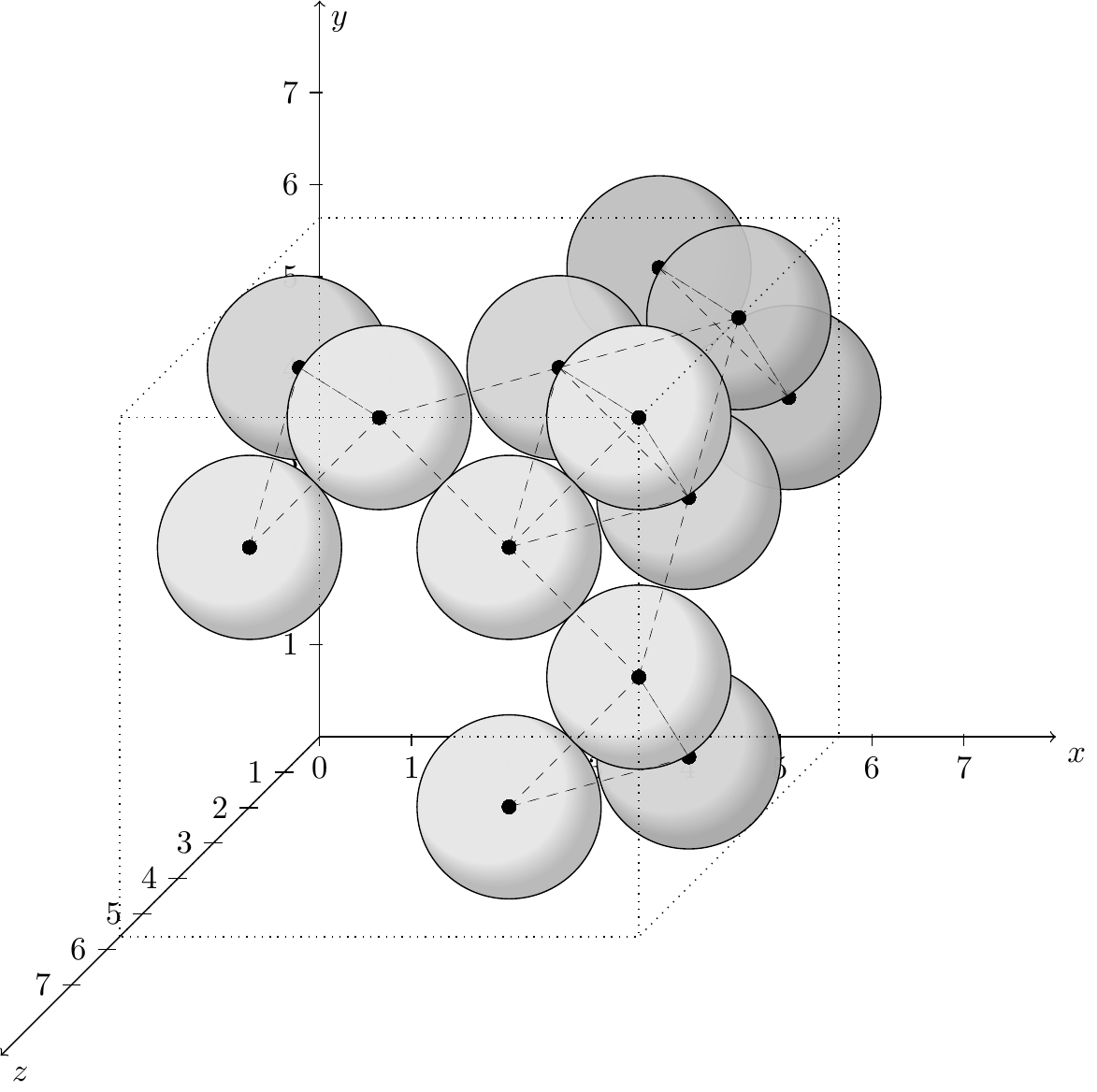}
  \caption{$L_{5,2}$}
  \vspace*{10mm}
\end{subfigure}%
\begin{subfigure}{.5\textwidth}
  \centering
  \includegraphics[width=.9\linewidth]{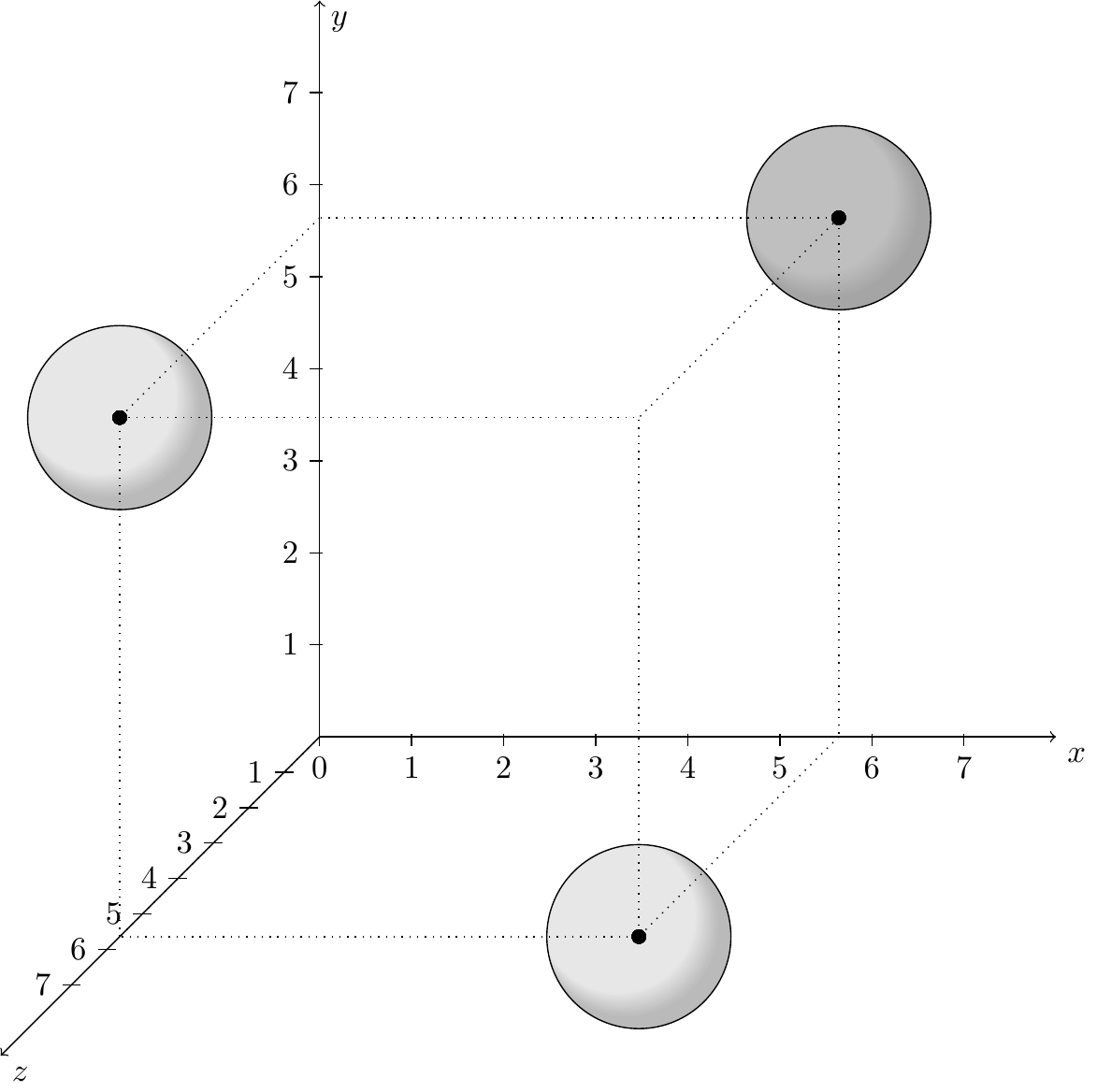}
  \caption{$L_{5,3}$}
  \vspace*{10mm}
\end{subfigure}
\caption{An illustration of arrangements $L_{5}$, $L_{5,1}$, $L_{5,2}$ and $L_{5,3}$.}
\end{figure}

Let $T_{p}=\tau(L_{p})$ be a set of points created by the translation
of points from $L_{p}$ by the function $\tau$ such that

\begin{equation}
T_{p}\in\mathcal{F},T_{p}\subset C_{p},h(P_{p-1},T_{p})\geq0\label{eq:tp_validated},
\end{equation}
\begin{equation}
h(T_{p},L_{p+1})>0\label{eq:tp_separated}
\end{equation}

\noindent for all $p>2$. If such a set $T_{p}$ exists, then we can state that
\[
P_{p}=P_{p-1}\cup T_{p}\text{, for all }p>2.
\]

We can easily show that $P_{2}$ exists by constructing it explicitly.
If we want to give $P_{2}$ in such a way to maximize the capability
to translate points from $L_{3}$ and thus produce a better packing $P_{3}$,
we have to find the positive root of the polynomial
\[
a^{4}+4a^{3}+8a^{2}-8=0,a>0\Longrightarrow a=0.818425\ldots
\]

\noindent then $P_{2}$ is given as $P_{2}=\left\{ \langle0,0,a\rangle,\langle b,b,0\rangle\right\}$,
where $b=\sqrt{2-a^{2}/2}$.

While $P_{p-1}$ and $L_{p}$ are separated (\ref{eq:p_separation})
, we can try to translate all points from $L_{p}$ and keep (\ref{eq:tp_validated}),
and (\ref{eq:tp_separated}) true. We split $L_{p}$ into three subsets,
and perform the translations on each of them while maintaining the
given conditions.

Let $L_{p}=L_{p,1}\cup L_{p,2}\cup L_{p,3}$ such that
\[
L_{p,1}=\left\{ \langle l_{1},l_{2},l_{3}\rangle\in L_{p}:l_{i},l_{j}\leq(p-3)\sqrt{2}\text{ for some distinct }i,j\in\left\{ 1,2,3\right\}\right\},
\]
\[
L_{p,2}=\left\{ \langle l_{1},l_{2},l_{3}\rangle\in L_{p}\setminus L_{p,1}:l_{1},l_{2},l_{3}>0\right\},
\]
\[
L_{p,3}=L_{p}\setminus L_{p,1}\setminus L_{p,2}.
\]

For an illustration of sets $L_{5}$, $L_{5,1}$, $L_{5,2}$ and $L_{5,3}$, see Figure 2.

Now we can give $T_{p}$ as
\[
T_{p}=T_{p,1}\cup T_{p,2}\cup T_{p,3},
\]
\[
T_{p,i}=\left\{ \langle l_{1}-u_{i,1}\tau_{i}(p),l_{2}-u_{i,2}\tau_{i}(p),l_{3}-u_{i,3}\tau_{i}(p)\rangle:\langle l_{1},l_{2},l_{3}\rangle\in L_{p,i}\right\},
\]
\[
u_{1,j}=\begin{cases}
1 & \text{if }l_{j}=\left(p-1\right)\sqrt{2}\\
0 & \text{otherwise}
\end{cases},u_{2,j}=1,u_{3,j}=\begin{cases}
1 & \text{if }l_{j}\neq0\\
0 & \text{otherwise}
\end{cases},
\]
where $\tau_{1}(p)$, $\tau_{2}(p)$ and $\tau_{3}(p)$ are small
numbers such that $\tau_{1}(p)>\tau_{2}(p)>\tau_{3}(p)$ and conditions
(\ref{eq:tp_validated}) and (\ref{eq:tp_separated}) hold. If we
additionally state that $h(P_{p-1},T_{p,1})=0$, $h(P_{p-1}\cup T_{p,1},T_{p,2})=0$
and $h(P_{p-1}\cup T_{p,1}\cup T_{p,2},T_{p,3})=0$, we can give explicit
solutions for $\tau_{1}$, $\tau_{2}$ and $\tau_{3}$ as follows
\[
\tau_{1}(3)=2\sqrt{2}-2-a,
\]
\[
\tau_{2}(p)=\frac{\sqrt{2}}{3}\left(\frac{\tau_{1}(p)}{\sqrt{2}}+2-\sqrt{-\tau_{1}(p)^{2}+2\sqrt{2}\tau_{1}(p)+4}\right)\text{, for all }p\geq3,
\]
\[
\tau_{3}(p)=\frac{\sqrt{2}}{2}\left(\sqrt{2}\tau_{2}(p)+1-\sqrt{-\tau_{2}(p)^{2}+2\sqrt{2}\tau_{2}(p)+1}\right)\text{, for all }p\geq3,
\]
\[
\tau_{1}(p)=\sqrt{2}+\tau_{3}(p-1)-\sqrt{-\tau_{3}(p-1)^{2}+2\sqrt{2}\tau_{3}(p-1)+2}\text{, for all }p>3.
\]

We notice that $D(P_{p})=D_{p}-\tau_{3}(p)$, hence the maximum separation
distance can be given as
\[
d_{n}\geq\frac{2}{\left(p-1\right)\sqrt{2}-\tau_{3}(p)}\text{, for all }n\leq g(p)-2
\]
and while $\tau_{3}(p)>0$, then $d_{n}>\sqrt{2}/\!(p-1)$.
\end{proof}

We denote the lower bound of improvements by
\[
I_{p}=d_{g(p)-2}-d'_{p}.
\]

By performing the calculations, we get particular values such as $I_{3}>8.235\cdot10^{-11}$,
$I_{4}>1.276\cdot10^{-79},\ldots$. 
This is a rough approximation while condition (\ref{eq:tp_separated})
needs to remain true for all $p$ and does not allow us to further improve
the packing for particular $p$. In the next section, we show
that $I_{p}$ is usually above this approximation.

We did not find a way to improve the packings when $n=g(p)-1$, and
we conjecture that in this case $d_{n}=d'_{p}$.

%%%%%%%%%%%%%%%%%%%%%%%%%%%%%%

\section{Optimization Approach}

Most of the existing packing methods focus on searching for a completely new arrangement of spheres,
usually performing a search from a randomly given initial position of spheres.
Such approach assumes that the packing of higher density can be reached
after a certain number of iterations and multiple runs of the search procedure using different
initial parameters \cite{gensane,wenqi,locatelli_de}.
The large number of iterations often limits the search procedure to the use of
double or quadruple floating-point precision, to maintain computation speed.
This precision is insufficient to detect improvements in many packings. Many of these approaches are
adapted and modified from widely known procedures for packing congruent
circles in a square or a circle \cite{csq_graham,cic_locatelli,csq_markot,csq_markot_2,csq_book}.

The method we suggest is based on a hypothesis that an improved packing
can be reached just by the omission of two or more spheres from the
ccp and by performing a translation of spheres using the available
space made after we remove the spheres.

The initial positions of the sphere centers we denote by $S_{p,r}\subset G_{p}$,
as a set of $g(p)-r$ points such that at least one sphere with a
center $s_{i}\in S_{p,r}$ can be continuously translated inside a
cube container without overlapping with other spheres. More precisely,

\[
\bigcup_{i=1}^{g(p)-r}\left\{ q\in C_{p}:q\notin S_{p,r},(S_{p,r}\setminus\left\{ s_{i}\right\} )\cup\left\{ q\right\} \in\mathfrak{\mathcal{\mathcal{F}}}\right\} \neq\varnothing.
\]

We can see that the construction of $S_{p,r}$ is possible for $r\geq2$,
and we experimentally determine solutions based on improvements
reached for certain arrangements. To simplify the search procedure,
instead of trying to figure out the best performing arrangements $S_{p,r}$
for each pair $(p,r)$, we find removal patterns $R_{r}=G_{p}\setminus S_{p,r}$
and use them to search for improvements for any $p$.

Table 1 shows patterns in a simplified notation where $R'_{r}=\left\{ \frac{\sqrt{2}}{2}s:s\in R_{r}\right\} $.
We notice that only $R_{2}$ is the most likely an optimal pattern for
all $p>1$.

\begin{table}[H]
\begin{centering}
\begin{tabular}{cl}
$r$ & $R'_{r}$\tabularnewline
\hline 
2 & $\left\{ \langle0,1,1\rangle,\langle1,1,0\rangle\right\} $\tabularnewline
3 & $\left\{ \langle0,0,0\rangle,\langle1,0,1\rangle,\langle2,0,0\rangle\right\} $\tabularnewline
4 & $\left\{ \langle0,1,1\rangle,\langle1,1,0\rangle,\langle1,1,2\rangle,\langle2,1,1\rangle\right\} $\tabularnewline
5 & $\left\{ \langle0,1,1\rangle,\langle1,0,1\rangle,\langle1,1,0\rangle,\langle2,0,0\rangle,\langle2,1,1\rangle\right\} $\tabularnewline
6 & $\left\{ \langle0,0,0\rangle,\langle0,1,1\rangle,\langle1,0,1\rangle,\langle1,1,0\rangle,\langle2,0,0\rangle,\langle2,1,1\rangle\right\} $\tabularnewline
\hline 
\end{tabular}
\par\end{centering}

\caption{Experimentally determined patterns $R'_{r}$}

\end{table}

Using the initial arrangement $S_{p,r}$ we try to perform
the translation of each sphere using the limited set of translation
vectors denoted by $T$. This algorithm can be described
as follows: 

\newpage
%\vspace{-4mm}
\begin{itemize}
\item For a given initial set $S_{n}\gets S_{p,r}$ repeat until $D(S_{n})<D_{p}$:

\begin{itemize}
\item \vspace{-3mm} For each $s_{i}\in S_{n}$ do:

\begin{itemize}
\item \vspace{-1mm}Randomly choose $t\in T$,
\item Let $v=\left\{ s_{i}+kt:k\in\left[-1,1\right]\right\}$,
\item Let $v_{i}=\left\{ q\in v:(S_{n}\setminus\left\{ s_{i}\right\} )\cup\left\{ q\right\} \in\mathcal{F},D(S_{n}\cup\{q\})=D(S_{n})\right\}$,
\item Find endpoints $a$ and $b$ of the largest line segment $\overline{ab}\subseteq v_{i}$ such that $s_{i}\in\overline{ab}$,
\item New position of $s_{i}$ is given as $\ensuremath{s_{i}\gets(a+b)/2}$.
\end{itemize}
\end{itemize}
\end{itemize}

After we try to move all points from $S_{n}$, we say that we completed one iteration.
Because of the very limited space to which we translate the spheres, and in order to
minimize the number of required iterations, we usually set 
$T=\left\{ 0,1\right\} ^{3}\setminus\left\{ \langle0,0,0\rangle\right\} $.
We also tested performances when $T$ takes different values such
as $\left\{ -1,0,1\right\} ^{3}$, more or less reduced sets, but
the improvements gained were always slightly worse.

In practice, $D(S_{p,r})$ is slightly larger than $D_{p}$ while
coordinates are given with finite precision. The described procedure
stops when the first improvement is detected, but if we continue the
search, we can improve the packing even more. It is also important
to set the precision above the expected value of $D_{p}-D(S_{p,r})$,
otherwise the improvement cannot be registered. Choosing the higher precision
enables us to reach an improvement with less iterations, but only
up to a certain level. We usually set the precision 1.5 times higher
than the expected improvement. If the precision is too high, the search
can be very slow, thus often we have to guess the range of possible
improvements using a lower precision at first, and increase it if
an improvement cannot be reached.

This approach is different from the procedures used for sphere packing
in the past, as in \cite{gensane,wenqi,locatelli_de,packomania},
while we focus only on tiny changes/improvements in the high density
structure. We cannot consider this approach as a good general packing
method for $r\geq p$. Its main weakness is that, because of the small
available space where spheres can be moved, random perturbations are
hard to implement, or at least we did not find any good method to
do it. Still, this method allows us to find improved packings in less
than one second for some well examined cases even when high precision
is not required, as for example $n=29$, 59 or 60.

The improvements attained are shown in the Tables 2 and 3. 
Table 2
lists values obtained for $I_{p}$ with $4\leq p\leq21$. Because of the
slow computation times for $p\geq13$, we ran a search with approximately
5000 iterations when improvement gains started to slow down. Table
3 lists other improved packings for $3\leq r\leq6$, $r<p$ and $4\leq p\leq12$
using $R_{r}$ patterns described in Table 1. The results are listed
as the best known values performed after a large number of
iterations and multiple runs of the search procedure.

%%%%%%%%%%%%%%%%%%%%%%%%%%%%%%

\begin{table}[H]
\centering{}\textsuperscript{}%
\begin{tabular}{lll}
$n$ & $p$ & \quad $I_{p}$\tabularnewline
\toprule
30 & 4 & $7.34\cdot10^{-68}$\tabularnewline
61 & 5 & $7.18\cdot10^{-80}$\tabularnewline
106 & 6 & $2.26\cdot10^{-314}$\tabularnewline
170 & 7 & $9.09\cdot10^{-622}$\tabularnewline
254 & 8 & $3.74\cdot10^{-629}$\tabularnewline
363 & 9 & $7.51\cdot10^{-629}$\tabularnewline
498 & 10 & $5.00\cdot10^{-2584}$\tabularnewline
664 & 11 & $9.67\cdot10^{-2563}$\tabularnewline
862 & 12 & $1.76\cdot10^{-4988}$\tabularnewline
1097 & 13 & $1.70\cdot10^{-5020}$\tabularnewline
1370 & 14 & $2.01\cdot10^{-5044}$\tabularnewline
1686 & 15 & $1.29\cdot10^{-5076}$\tabularnewline
2046 & 16 & $1.78\cdot10^{-5116}$\tabularnewline
2455 & 17 & $3.30\cdot10^{-5047}$\tabularnewline
2914 & 18 & $1.47\cdot10^{-10118}$\tabularnewline
3428 & 19 & $4.20\cdot10^{-10121}$\tabularnewline
3998 & 20 & $1.16\cdot10^{-20344}$\tabularnewline
4629 & 21 & $2.46\cdot10^{-20582}$\tabularnewline
\bottomrule
\end{tabular}\caption{Improved values of $I_{p}$}
\end{table}

%\pagebreak{}
\vspace*{-\baselineskip}

\begin{table}[H]
\begin{centering}
\begin{tabular}{llll|llll}
$n$ & $p$ & $r$ & $d_{n}-d'_{p}$ & $n$ & $p$ & $r$ & $d_n - d'_p$  \\
\hline
29 & 4 & 3 & $2.23\cdot10^{-12}$ & & & & \\
\hline
59 & 5 & 4 & $1.95\cdot10^{-11}$ & & & & \\
60 & 5 & 3 & $2.09\cdot10^{-20}$ & & & & \\
\hline
103 & 6 & 5 & $3.38\cdot10^{-14}$ & & & & \\
104 & 6 & 4 & $9.98\cdot10^{-47}$ & & & & \\
105 & 6 & 3 & $1.34\cdot10^{-76}$ & & & & \\
\hline
166 & 7 & 6 & $3.08\cdot10^{-21}$ & 494 & 10 & 6 & $4.77\cdot10^{-57}$\\
167 & 7 & 5 & $7.72\cdot10^{-31}$ & 495 & 10 & 5 & $6.96\cdot10^{-199}$ \\
168 & 7 & 4 & $1.59\cdot10^{-87}$ & 496 & 10 & 4 & $1.72\cdot10^{-310}$ \\
169 & 7 & 3 & $8.49\cdot10^{-148}$ & 497 & 10 & 3 & $1.28\cdot10^{-605}$ \\
\hline
250 & 8 & 6 & $2.98\cdot10^{-28}$ & 660 & 11 & 6 & $2.83\cdot10^{-119}$\\
251 & 8 & 5 & $1.99\cdot10^{-43}$ & 661 & 11 & 5 & $1.99\cdot10^{-173}$\\
252 & 8 & 4 & $3.11\cdot10^{-102}$ & 662 & 11 & 4 & $1.40\cdot10^{-407}$\\
253 & 8 & 3 & $6.02\cdot10^{-153}$ & 663 & 11 & 3 & $3.47\cdot10^{-615}$\\
\hline
359 & 9 & 6 & $3.84\cdot10^{-28}$ & 858 & 12 & 6 & $8.40\cdot10^{-248}$\\
360 & 9 & 5 & $6.45\cdot10^{-44}$ & 859 & 12 & 5 & $2.68\cdot10^{-404}$\\
361 & 9 & 4 & $2.05\cdot10^{-101}$ & 860 & 12 & 4 & $7.74\cdot10^{-745}$\\
362 & 8 & 3 & $1.08\cdot10^{-152}$ & 861 & 12 & 3 & $2.91\cdot10^{-1212}$\\
\hline
\end{tabular}
\par\end{centering}
\caption{Improvements reached using patterns $R_{r}$}
\end{table}

\bibliographystyle{plain}

\end{document}